\documentclass[proof]{WileyASNA-v1}
\usepackage{tensor}
\usepackage{verbatim}
\newcommand{\RB}{\mathbb{R}}
\newcommand{\ZB}{\mathbb{Z}}
\newcommand{\pfrac}[2]{\frac{\partial{#1}}{\partial{#2}}}
\newcommand{\HCd}{\mathcal{H}}
\newcommand{\LCd}{\mathcal{L}}
\newcommand{\onehalf}{{\textstyle\frac{1}{2}}}
\newcommand{\quarter}{{\textstyle\frac{1}{4}}}
\newcommand{\bgamma}{\pmb{\gamma}}
\newcommand{\bsigma}{\pmb{\sigma}}
\newcommand{\rmi}{\mathrm{i}}

\newcommand{\dete}{\varepsilon}
\DeclareMathOperator{\Tr}{Tr}
\articletype{Proceeding}%

\received{04 October 2022}
\revised{08 October 2022}
\accepted{09 October 2022}
\doi{10.1002/asna.20220074}
\begin{document}

\title{Identity for scalar-valued functions of tensors and its applications\\
\subtitlefont{to energy-momentum tensors in classical field theories and gravity}}

\author[1,2]{J\"urgen Struckmeier}
\author[1,2]{Armin van de Venn}
\author[1]{David Vasak}

\authormark{STRUCKMEIER \textsc{et al}}

\address[1]{\orgname{Frankfurt Institute for Advanced Studies (FIAS)}, \orgaddress{Ruth-Moufang-Str. 1,\\ 60438 Frankfurt am Main}, \country{Germany}}
\address[2]{\orgname{Goethe Universit\"at}, \orgaddress{Max-von-Laue-Str. 1, 60438 Frankfurt am Main}, \country{Germany}}
\corres{Jürgen Struckmeier, Frankfurt Institute for Advanced Studies (FIAS), Ruth-Moufang-Str. 1, 60438 Frankfurt am Main, Germany.\\
\email{struckmeier@fias.uni-frankfurt.de}}

\abstract{
We prove a theorem on scalar-valued functions of tensors, where ``scalar'' refers to absolute scalars as well as relative scalars of weight $w$.
The present work thereby generalizes an identity referred to earlier by Rosenfeld in his publication ``On the energy-momentum tensor''.
The theorem provides a $(1,1)$-tensor identity which can be regarded as the tensor analogue of the identity following from Euler's theorem on homogeneous functions.
The remarkably simple identity is independent of any internal symmetries of the constituent tensors,
providing a powerful tool for deriving relations between field-theoretical expressions and physical quantities.
We apply the identity especially for analyzing the metric and canonical energy-momentum tensors of matter and gravity and the relation between them.
Moreover, we present a generalized Einstein field equation for arbitrary version of vacuum space-time dynamics --- including torsion and non-metricity.
The identity allows to formulate an equivalent representation of this equation.
Thereby the conjecture of a zero-energy universe is confirmed.
}

\keywords{Scalar-valued function of tensors; Energy-momentum tensor; Semi-classical field theory of gravity and matter;
Generalized Einstein field equation; Zero-energy universe.}


\fundingInfo{Walter Greiner Gesellschaft zur Förderung der physikalischen Grundlagenforschung e.V., Frankfurt am Main}

\maketitle


\section{Introduction}

In this paper we prove an identity for scalar-valued functions of tensors that constitutes the tensor analogue of Euler's theorem on homogeneous functions of calculus.
Its proof is rather straightforward and the identity itself appears almost trivial, but certainly not obvious.
The identity provides a valuable and effective means for establishing relations between tensor expressions.
It remained rather unnoticed by the community.

A prominent example addressed here is the long-standing question of the permissibility of
``improving'' the canonical energy-momentum tensor for fields of spin higher than zero.
This topic was addressed much earlier by Rosenfeld~\citep{rosenfeld40} and by Belinfante~\citep{belinfante39}
in the context of symmetrizing non-symmetric canonical energy-momentum tensors.
The common objective was to identify a form suitable for furnishing the source term in the classical Einstein field equation with its symmetric Einstein tensor.
In later papers of Sciama~\citep{sciama62}, Kibble~\citep{kibble61}, and Hehl~et~al.~\citep{hehl76}, non-symmetric canonical energy-momentum tensors were
identified as the source of torsion of space-time in a generalized Einstein-Cartan-Sciama-Kibble (ECSK) theory of gravity.

The paper is organized as follows.
In Section~\ref{sec:proof} we first prove the theorem for \emph{absolute} scalars constituted
by contracting $(n,n)$-tensors that on their part may be tensor products of lower rank tensors.
On this basis, the theorem is then generalized for tensors whose indices refer to different vector spaces,
such as vierbeins (tetrads) and spinor-tensors (e.g.\ Dirac matrices), and to \emph{relative} tensors of any weight $w$.
As an instructive example, we apply the theorem to the determinant of a $(0,2)$-tensor, which represents a relative scalar of weight $w=2$,
We thus find immediately the derivative of this determinant with respect to a tensor component---without the need to refer to the definition of a determinant.

In Section~\ref{sec:applications} we apply the identity to Lagrangians of semi-classical field theories that describe the dynamics
of scalar, vector, and spinor fields in an arbitrary geometry of space-time.
These Lagrangians are relative scalars of weight $w=1$, commonly referred to as scalar densities.
While the functional derivative of a Lagrangian density with respect to the metric gives the \emph{metric} energy-momentum tensor,
the \emph{canonical} one is derived from Noether's theorem.
With help of the identity, a generic relation between the metric and the canonical (Noether) energy-momentum tensors
is then established, and applied to the Klein-Gordon, Proca, and Dirac systems.

In Section~\ref{sec:gravity} the identity is applied to a variety of gravity theories, some including torsion and even non-metricity,
with the metric and canonical energy-momentum tensor derived in analogy to the treatment of matter systems.
This analogy naturally leads to a conservation law of the total energy-momentum of matter and space-time (``Zero-Energy Universe'').
The underlying energy-momentum balance equation is identified as the generalized \emph{consistency equation}
that in the specific case of the Einstein-Hilbert Lagrangian gives the well-known field equation of Einstein's General Relativity.
In the Riemann-Cartan geometry a Poisson-like equation for torsion emerges in addition to an extended field equation.

Section~\ref{sec:conclusions} finally gives the summary of the paper and the conclusions.
\section{Theorem for scalar-valued functions of tensors}\label{sec:proof}
\subsection{Absolute tensors}
In the following, we prove a theorem for a scalar that is formed from contractions of an arbitrary $(n,n)$ tensor $T$
and for any $(n,n)$ tensor field $T(x)$ at a point with the coordinate $x$ in an arbitrary geometry of the underlying manifold.
\footnote{A similar consideration can be found in Rosenfeld's paper~\citep{rosenfeld40}.}
\begin{theorem}
Let $S=S(T)\in\RB$ be a scalar-valued function constructed from a complete contraction of a $(n,n)$-tensor
$T\indices{^{\xi_{1}\ldots\xi_{n}}_{\eta_{1}\ldots\eta_{n}}}$, with the ordered index set $\{\eta_k\}$ a bijective
permutation $\{\xi_{\pi(k)}\}$ of the ordered index set $\{\xi_j\}$ such that $\eta_k\equiv\xi_{\pi(k)}\equiv\xi_j$.
Then the following identity holds:
\begin{align}
&-\pfrac{S}{T\indices{^{\nu \xi_{2}\ldots\xi_{n}}_{ \eta_{1}\ldots\eta_{n}}}}
T\indices{^{\mu \xi_{2}\ldots\xi_{n}}_{ \eta_{1}\ldots\eta_{n}}}\nonumber\\
&-\ldots-\pfrac{S}{T\indices{^{\xi_{1}\ldots\xi_{n-1} \nu}_{ \eta_{1}\ldots\eta_{n}}}}
T\indices{^{\xi_{1}\ldots\xi_{n-1} \mu}_{ \eta_{1}\ldots\eta_{n}}}\nonumber\\
&+\pfrac{S}{T\indices{^{\xi_{1}\ldots\xi_{n}}_{ \mu \eta_{2}\ldots\eta_{n}}}}
T\indices{^{\xi_{1}\ldots\xi_{n}}_{ \nu \eta_{2}\ldots\eta_{n}}}\nonumber\\
&+\ldots+\pfrac{S}{T\indices{^{\xi_{1}\ldots\xi_{n}}_{ \eta_{1}\ldots\eta_{n-1} \mu}}}
T\indices{^{\xi_{1}\ldots\xi_{n}}_{ \eta_{1}\ldots\eta_{n-1} \nu}}\equiv0.
\label{eq:assertion1}
\end{align}
Thus, in each term, exactly one previously contracted index $\xi_j$ resp.\ $\eta_j$ of $S$
is replaced by the now open indices $\nu$ and $\mu$.
\end{theorem}
\begin{proof}
The proof is based on the simple relation
\begin{align*}
\pfrac{S}{T\indices{^\alpha_\mu}}T\indices{^\alpha_\nu}&=
\pfrac{T\indices{^\beta_\beta}}{T\indices{^\alpha_\mu}}T\indices{^\alpha_\nu}=
\delta_\alpha^\beta \delta_\beta^\mu T\indices{^\alpha_\nu}=T\indices{^\mu_\nu}\\
&=\delta_\nu^\beta \delta_\beta^\alpha T\indices{^\mu_\alpha}=
\pfrac{T\indices{^\beta_\beta}}{T\indices{^\nu_\alpha}}T\indices{^\mu_\alpha}=
\pfrac{S}{T\indices{^\nu_\alpha}}T\indices{^\mu_\alpha}
\end{align*}
It is straightforward to see that this relation applies for any pair of contracted indices.
Then for any pair $\xi_j$ and $\eta_k$ of contracted indices $\eta_k\equiv\xi_{\pi(k)}\equiv\xi_j$ of $S=T\indices{^{\xi_{1}\ldots\xi_{n}}_{\eta_{1}\ldots\eta_{n}}}$
we find for their respective replacement by $\nu$ and $\mu$ according to Eq.~(\ref{eq:assertion1}):
\begin{align*}
&-\pfrac{S}{T\indices{^{\xi_{1}\ldots\xi_{j-1} \nu \xi_{j+1}\ldots\xi_{n}}_{\eta_1\ldots\eta_{k-1} \eta_{k} \eta_{k+1}\ldots\eta_{n}}}}
T\indices{^{\xi_{1}\ldots\xi_{j-1} \mu \xi_{j+1}\ldots\xi_{n}}_{\eta_1\ldots\eta_{k-1} \eta_{k} \eta_{k+1}\ldots\eta_{n}}}\\
&+\pfrac{S}{T\indices{^{\xi_{1}\ldots\xi_{j-1} \xi_j \xi_{j+1}\ldots\xi_{n}}_{\eta_1\ldots\eta_{k-1} \mu \eta_{k+1}\ldots\eta_{n}}}}
T\indices{^{\xi_{1}\ldots\xi_{j-1} \xi_j \xi_{j+1}\ldots\xi_{n}}_{\eta_{1}\ldots\eta_{k-1} \nu \eta_{k+1}\ldots\eta_{n}}}\\
&=-\delta_{\nu}^{\xi_j} T\indices{^{\xi_{1}\ldots\xi_{j-1} \mu \xi_{j+1}\ldots\xi_{n}}_{\eta_{1}\ldots\eta_{k-1} \eta_{k} \eta_{k+1}\ldots\eta_{n}}}\\
&\quad+\delta_{\eta_{k}}^{\mu} T\indices{^{\xi_{1}\ldots\xi_{j-1} \xi_j \xi_{j+1}\ldots\xi_{n}}_{\eta_{1}\ldots\eta_{k-1} \nu \eta_{k+1}\ldots\eta_{n}}}\\
&=-\delta_{\nu}^{\xi_j} T\indices{^{\xi_{1}\ldots\xi_{j-1} \mu \xi_{j+1}\ldots\xi_{n}}_{\eta_{1}\ldots\eta_{k-1} \xi_{j} \eta_{k+1}\ldots\eta_{n}}}\\
&\quad+\delta_{\eta_{k}}^{\mu} T\indices{^{\xi_{1}\ldots\xi_{j-1} \eta_k \xi_{j+1}\ldots\xi_{n}}_{\eta_{1}\ldots\eta_{k-1} \nu \eta_{k+1}\ldots\eta_{n}}}\\
&=-T\indices{^{\xi_{1}\ldots\xi_{j-1} \mu \xi_{j+1}\ldots\xi_{n}}_{\eta_{1}\ldots\eta_{k-1} \nu \eta_{k+1}\ldots\eta_{n}}}\\
&\quad+T\indices{^{\xi_{1}\ldots\xi_{j-1} \mu \xi_{j+1}\ldots\xi_{n}}_{\eta_{1}\ldots\eta_{k-1} \nu \eta_{k+1}\ldots\eta_{n}}}\\
&\equiv0,
\end{align*}
which proves the assertion.
\end{proof}
The $(n,n)$-tensor $T\indices{^{\xi_{1}\ldots\xi_{n}}_{ \eta_{1}\ldots\eta_{n}}}$
may be in particular a tensor product of tensors of lower ranks.
This is demonstrated in the following example:
\begin{eexample}
Let $S$ be a scalar emerging from the contraction of arbitrary tensors
$A\indices{^\alpha^\beta_\xi}$ with $B\indices{_\alpha^\xi^\eta}$ and $C_\beta$:
\begin{equation}\label{eq:arb-scal-def}
S=A\indices{^\alpha^\beta_\xi} B\indices{_\alpha^\xi^\eta} C_\beta C_\eta.
\end{equation}
Then
\begin{align}
&-\pfrac{S}{A\indices{^{\nu\beta}_{\xi}}}A\indices{^{\mu\beta}_{\xi}}
-\pfrac{S}{A\indices{^{\alpha\nu}_{\xi}}}A\indices{^{\alpha\mu}_{\xi}}
+\pfrac{S}{A\indices{^{\alpha\beta}_{\mu}}}A\indices{^{\alpha\beta}_{\nu}}\nonumber\\
&-\pfrac{S}{B\indices{_\alpha^\nu^\eta}}B\indices{_\alpha^\mu^\eta}
-\pfrac{S}{B\indices{_\alpha^\xi^\nu}}B\indices{_\alpha^\xi^\mu}
+\pfrac{S}{B\indices{_\mu^\xi^\eta}}B\indices{_\nu^\xi^\eta}+\pfrac{S}{C_\mu}C_\nu\equiv0.
\label{eq:arb-scal}
\end{align}
For the proof of the identity~(\ref{eq:arb-scal}) the respective terms of the sum are worked out explicitly:
\begin{align*}
-\pfrac{S}{A\indices{^{\nu\beta}_\xi}}A\indices{^{\mu\beta}_\xi}&=
-\delta_\nu^\alpha A\indices{^{\mu\beta}_\xi} B\indices{_\alpha^\xi^\eta} C_\beta C_\eta
=-A\indices{^{\mu\beta}_\xi} B\indices{_\nu^\xi^\eta} C_\beta C_\eta\\
-\pfrac{S}{A\indices{^{\alpha\nu}_\xi}}A\indices{^{\alpha\mu}_\xi}&=
-\delta_\nu^\beta A\indices{^{\alpha\mu}_\xi} B\indices{_\alpha^\xi^\eta} C_\beta C_\eta
=-A\indices{^{\alpha\mu}_\xi} B\indices{_\alpha^\xi^\eta} C_\nu C_\eta\\
\hphantom{-}\pfrac{S}{A\indices{^{\alpha\beta}_\mu}}A\indices{^{\alpha\beta}_\nu}&=
\hphantom{-}\delta_\xi^\mu A\indices{^{\alpha\beta}_\nu} B\indices{_\alpha^\xi^\eta} C_\beta C_\eta
=\hphantom{-}A\indices{^{\alpha\beta}_\nu} B\indices{_\alpha^\mu^\eta} C_\beta C_\eta\\
-\pfrac{S}{B\indices{_\alpha^\nu^\eta}}B\indices{_\alpha^\mu^\eta}&=
-A\indices{^{\alpha\beta}_\xi} \delta_\nu^\xi B\indices{_\alpha^\mu^\eta} C_\beta C_\eta
=-A\indices{^{\alpha\beta}_\nu} B\indices{_\alpha^\mu^\eta} C_\beta C_\eta\\
-\pfrac{S}{B\indices{_\alpha^{\xi\nu}}}B\indices{_\alpha^{\xi\mu}}
&=-A\indices{^{\alpha\beta}_\xi} \delta_\nu^\eta B\indices{_\alpha^\xi^\mu} C_\beta C_\eta
=-A\indices{^{\alpha\beta}_\xi} B\indices{_\alpha^\xi^\mu} C_\beta C_\nu\\
\hphantom{-}\pfrac{S}{B\indices{_\mu^{\xi\eta}}}B\indices{_\nu^{\xi\eta}}
&=\hphantom{-}A\indices{^{\alpha\beta}_\xi} \delta_\alpha^\mu B\indices{_\nu^\xi^\eta} C_\beta C_\eta
=\hphantom{-}A\indices{^{\mu\beta}_\xi} B\indices{_\nu^\xi^\eta} C_\beta C_\eta\\
\hphantom{-}\pfrac{S}{C_\mu}C_\nu
&=\hphantom{-}A\indices{^{\alpha\beta}_\xi} B\indices{_\alpha^\xi^\eta} \delta^\mu_\beta C_\nu C_\eta
+A\indices{^{\alpha\beta}_\xi} B\indices{_\alpha^\xi^\eta} C_\beta \delta_\eta^\mu C_\nu\\
&=\hphantom{-}A\indices{^{\alpha\mu}_\xi} B\indices{_\alpha^\xi^\eta} C_\nu C_\eta
+A\indices{^{\alpha\beta}_\xi} B\indices{_\alpha^\xi^\mu} C_\beta C_\nu.
\end{align*}
The terms on the right-hand sides obviously sum up to zero.
\end{eexample}
Of course, this example can be generalized to the contraction of any number of tensors of any rank.
It is important to stress that the above identity is entirely unrelated to possible internal symmetries of the constituent tensors,
and is valid in addition to any symmetry related identities like e.g.\ the Bianchi identities in field theories of gravity involving the Riemann tensor.
Nonetheless, internal tensor symmetries can be exploited to simplify the resulting identity.
If, for instance, the tensor $A\indices{^\alpha^\beta_\xi}$ in~(\ref{eq:arb-scal-def}) is symmetric or skew-symmetric in its upper index pair,
then the first two terms on the left-hand side of the identity~(\ref{eq:arb-scal}) can be merged into a single term.
\subsection{Scalar functions involving the metric tensor}
For $S=S(g,T)\in\RB$, a scalar-valued function constructed from the metric tensor $g_{\mu\nu}$
and an $(n,m)$-tensor $T\indices{^{\xi_{1}\ldots\xi_{n}}_{\eta_{1}\ldots\eta_{m}}}$,
the following identity then holds for \mbox{$(n-m)/2\in\ZB$}:
\begin{align}
0\equiv\pfrac{S}{g_{\mu\beta}}g_{\nu\beta}&+\pfrac{S}{g_{\beta\mu}}g_{\beta\nu}-
\pfrac{S}{T\indices{^{\nu \xi_{2}\ldots\xi_{n}}_{\eta_{1}\ldots\eta_{m}}}}
T\indices{^{\mu \xi_{2}\ldots\xi_{n}}_{\eta_{1}\ldots\eta_{m}}}\nonumber\\
&-\ldots-\pfrac{S}{T\indices{^{\xi_{1}\ldots\xi_{n-1} \nu}_{\eta_{1}\ldots\eta_{m}}}}
T\indices{^{\xi_{1}\ldots\xi_{n-1} \mu}_{\eta_{1}\ldots\eta_{m}}}\nonumber\\
&+\pfrac{S}{T\indices{^{\xi_{1}\ldots\xi_{n}}_{\mu \eta_{2}\ldots\eta_{m}}}}
T\indices{^{\xi_{1}\ldots\xi_{n}}_{\nu \eta_{2}\ldots\eta_{m}}}\nonumber\\
&+\ldots+\pfrac{S}{T\indices{^{\xi_{1}\ldots\xi_{n}}_{\eta_{1}\ldots\eta_{m-1} \mu}}}
T\indices{^{\xi_{1}\ldots\xi_{n}}_{\eta_{1}\ldots\eta_{m-1} \nu}}.
\label{eq:assertion1a}
\end{align}
\begin{corollary}
The trace of Eq.~(\ref{eq:assertion1a}) then yields the scalar identity:
\begin{equation}\label{eq:assertion-contr}
\pfrac{S}{g_{\alpha\beta}}g_{\alpha\beta}\equiv\frac{n-m}{2}\pfrac{S}{T\indices{^{\xi_{1}\ldots\xi_{n}}_{\eta_{1}\ldots\eta_{m}}}}
T\indices{^{\xi_{1}\ldots\xi_{n}}_{\eta_{1}\ldots\eta_{m}}}.
\end{equation}
\end{corollary}
\begin{proof}
Contracting Eq.~(\ref{eq:assertion1a}) immediately gives Eq.~(\ref{eq:assertion-contr}).
\end{proof}
\subsection{Scalar functions involving tensors with multiple index classes}
\begin{corollary}
The theorem~(\ref{eq:assertion1}) holds also for scalars $S$ constructed from generalized tensor objects---such
as spinor-tensors---which are made of multiple index classes.
Examples of such objects are Dirac matrices $\gamma^\mu$, which are $(1,1)$-spinor-$(1,0)$-tensors, as well as vierbeins (tetrads)
with one Lorentz and one space-time index, respectively.
One then encounters a specific identity for each particular index class, provided that all other indices are fully contracted.
\end{corollary}
\begin{proof}
For each fully contracted index group, all terms on the left-hand side of Eq.~(\ref{eq:assertion1}) cancel.
For the not fully contracted indices, Eq.~(\ref{eq:assertion1}) applies.
\end{proof}
Examples are provided in Sects.~\ref{sec:dirac-all} and~\ref{sec:hilbert-tetr}.

\subsection{Relative scalar built from relative tensors}
\begin{corollary}
Let a \emph{relative scalar of weight} $w$---denoted by $\tilde{S}=S{\left(\sqrt{-g}\right)}^w\in\RB$---be
given as a function of the metric $g_{\mu\nu}$ and a tensor
$T\indices{^{\xi_{1}\ldots\xi_{n}}_{\eta_{1}\ldots\eta_{m}}}$ of rank $(n,m)$.
Then the following identity holds for $(n-m)/2\in\ZB$:
\begin{align}
&+\pfrac{\tilde{S}}{g_{\mu\beta}}g_{\nu\beta}+\pfrac{\tilde{S}}{g_{\beta\mu}}g_{\beta\nu}
-\pfrac{\tilde{S}}{T\indices{^{\nu \xi_{2}\ldots\xi_{n}}_{\eta_{1}\ldots\eta_{m}}}}
T\indices{^{\mu \xi_{2}\ldots\xi_{n}}_{\eta_{1}\ldots\eta_{m}}}\nonumber\\
&-\ldots-\pfrac{\tilde{S}}{T\indices{^{\xi_{1}\ldots\xi_{n-1} \nu}_{\eta_{1}\ldots\eta_{m}}}}
T\indices{^{\xi_{1}\ldots\xi_{n-1} \mu}_{\eta_{1}\ldots\eta_{m}}}\nonumber\\
&+\pfrac{\tilde{S}}{T\indices{^{\xi_{1}\ldots\xi_{n}}_{\mu \eta_{2}\ldots\eta_{m}}}}
T\indices{^{\xi_{1}\ldots\xi_{n}}_{\nu \eta_{2}\ldots\eta_{m}}}\nonumber\\
&+\ldots+\pfrac{\tilde{S}}{T\indices{^{\xi_{1}\ldots\xi_{n}}_{\eta_{1}\ldots\eta_{m-1} \mu}}}
T\indices{^{\xi_{1}\ldots\xi_{n}}_{\eta_{1}\ldots\eta_{m-1} \nu}}\equiv\delta_{\nu}^{\mu} w \tilde{S}.
\label{eq:assertion2}
\end{align}
\end{corollary}
\begin{proof}
Multiply Eq.~(\ref{eq:assertion1a}) with ${\left(\sqrt{-g}\right)}^w$, add $\delta_{\nu}^{\mu} w \tilde{S}$
on both sides of the identity, and combine the appropriate terms on the left-hand side.
From the definition of the determinant, the general rule for the derivative of the determinant of the covariant metric
with respect to a component of the metric is obtained as:
\begin{equation}\label{eq:det-deri}
\pfrac{\sqrt{-g}}{g_{\mu\beta}}=\onehalf g^{\beta\mu}\sqrt{-g}.
\end{equation}
One then gets:
\begin{align*}
&\hphantom{=}\left(\pfrac{S}{g_{\mu\beta}}g_{\nu\beta}+\pfrac{S}{g_{\beta\mu}}g_{\beta\nu}
+\delta_\nu^\mu w S\right){\left(\sqrt{-g}\right)}^w\\
&=\pfrac{\tilde{S}}{g_{\mu\beta}}g_{\nu\beta}+\pfrac{\tilde{S}}{g_{\beta\mu}}g_{\beta\nu}.
\end{align*}
\end{proof}
Equation~(\ref{eq:assertion2}) is obviously the analogue of Euler's theorem on homogeneous functions in the realm of tensor calculus.
The relative scalar $\tilde{S}=S{\left(\sqrt{-g}\right)}^w$ of weight $w$ may in particular be the tensor product
of some relative tensors of lower ranks and weights.
The weight $w$ of $\tilde{S}$ is then the \emph{sum} of the weights of the relative tensors.

As a direct application, the relation~(\ref{eq:det-deri}) for the derivative of
the metric's determinant may be recovered using \eqref{eq:assertion2}.
\begin{eexample}
The components of the covariant metric tensor field $g_{\mu\nu}(x)$ transform under the transition $x\mapsto X$ of the space-time coordinates as:
\begin{equation*}
g_{\mu\nu}(X)=\pfrac{x^\alpha}{X^\mu} g_{\alpha\beta}(x)\pfrac{x^\beta}{X^\nu},
\end{equation*}
which gives for the determinant $g\equiv\left|g_{\alpha\beta}(x)\right|$:
\begin{equation*}
\left|g_{\mu\nu}(X)\right|=\left|g_{\alpha\beta}(x)\right|{\left|\pfrac{x}{X}\right|}^2.
\end{equation*}
The determinant $g$ of the covariant metric tensor thus transforms as a \emph{relative scalar} of weight $w=2$.
According to the general form of the identity for relative scalars of weight $w$ from Eq.~(\ref{eq:assertion2}),
we get for the relative scalar $g$ due to the symmetry $g_{\beta\alpha}=g_{\alpha\beta}$:
\begin{equation}\label{eq:metric-identity}
\pfrac{g}{g_{\beta\mu}}g_{\beta\nu}+\pfrac{g}{g_{\mu\beta}}g_{\nu\beta}\equiv2 \delta_\nu^\mu g\quad\Rightarrow\quad
\pfrac{g}{g_{\beta\mu}}g_{\beta\alpha}\equiv\delta_\alpha^\mu g.
\end{equation}
Contracting~(\ref{eq:metric-identity}) with the inverse metric $g^{\alpha\nu}$ reproduces the derivative of the determinant
$g$ of the covariant metric with respect to the component $g_{\nu\mu}$ of the metric from Eq.~(\ref{eq:det-deri}),
\begin{equation*}
\pfrac{g}{g_{\nu\mu}}\equiv g^{\mu\nu} g,
\end{equation*}
and thus for the negative square root of $g$
\begin{equation}\label{eq:metric-identity2}
\pfrac{\sqrt{-g}}{g_{\nu\mu}}\equiv\onehalf g^{\mu\nu}\sqrt{-g},\qquad
\pfrac{\sqrt{-g}}{g^{\nu\mu}}\equiv-\onehalf g_{\mu\nu}\sqrt{-g}.
\end{equation}
\end{eexample}
\section{Physical applications}\label{sec:applications}
\subsection{Lagrangian densities}
In the following applications various relativistic field theories are considered where the underlying space-time is not restricted
to the Minkowski geometry but takes general curvilinear geometries into account.
In that case the invariance of the action integral requires that the Lagrangian is a scalar density rather than a density.
This is accomplished by multiplying it with $\sqrt{-g}$.
All tensors forming $\tilde{\LCd}$ are understood as \emph{tensor fields}, taken at the same space-time event~$x$.
\begin{corollary}
For a \emph{scalar density} Lagrangian $\tilde{\LCd}$, i.e.\ for a relative scalar of weight $w=1$,
the identity~(\ref{eq:assertion2}) gives for $(n-m)/2\in\ZB$:
\begin{align}
&-\pfrac{\tilde{\LCd}}{g^{\nu\beta}}g^{\mu\beta}-\pfrac{\tilde{\LCd}}{g^{\beta\nu}}g^{\beta\mu}
-\pfrac{\tilde{\LCd}}{T\indices{^{\nu \xi_{2}\ldots\xi_{n}}_{\eta_{1}\ldots\eta_{m}}}}
T\indices{^{\mu \xi_{2}\ldots\xi_{n}}_{\eta_{1}\ldots\eta_{m}}}\nonumber\\
&-\ldots-\pfrac{\tilde{\LCd}}{T\indices{^{\xi_{1}\ldots\xi_{n-1} \nu}_{\eta_{1}\ldots\eta_{m}}}}
T\indices{^{\xi_{1}\ldots\xi_{n-1} \mu}_{\eta_{1}\ldots\eta_{m}}}\nonumber\\
&+\pfrac{\tilde{\LCd}}{T\indices{^{\xi_{1}\ldots\xi_{n}}_{\mu \eta_{2}\ldots\eta_{m}}}}
T\indices{^{\xi_{1}\ldots\xi_{n}}_{\nu \eta_{2}\ldots\eta_{m}}}\nonumber\\
&+\ldots+\pfrac{\tilde{\LCd}}{T\indices{^{\xi_{1}\ldots\xi_{n}}_{\eta_{1}\ldots\eta_{m-1} \mu}}}
T\indices{^{\xi_{1}\ldots\xi_{n}}_{\eta_{1}\ldots\eta_{m-1} \nu}}
\equiv\delta_{\nu}^{\mu} \tilde{\LCd}.
\label{eq:assertion3}
\end{align}
\end{corollary}
\subsection{Klein-Gordon Lagrangian\label{ex:KG-ham}}
The Klein-Gordon Lagrangian density $\tilde{\LCd}_{\mathrm{KG}}=\LCd_{\mathrm{KG}}\sqrt{-g}$
for a massive \emph{complex} scalar field $\phi(x)$ is given by (see e.g.~\cite{greiner96}):
\begin{align}
\tilde{\LCd}_{\mathrm{KG}}&\left(\phi,\bar{\phi},\partial_\mu\phi,\partial_\nu\bar{\phi},g^{\mu\nu}\right)
=\left(\pfrac{\bar{\phi}}{x^{\alpha}}\pfrac{\phi}{x^{\beta}}g^{\alpha\beta}-m^2\bar{\phi} \phi\right)\sqrt{-g}\nonumber\\
&=\left[\frac{1}{2}\left(\pfrac{\bar{\phi}}{x^{\alpha}}\pfrac{\phi}{x^{\beta}}
+\pfrac{\bar{\phi}}{x^{\beta}}\pfrac{\phi}{x^{\alpha}}\right)g^{\alpha\beta}-m^2\bar{\phi} \phi\right]\sqrt{-g}.
\label{eq:KG-Lagrangian-complex}
\end{align}
In order to set up the pertaining identity~(\ref{eq:assertion2}), we first calculate the required derivatives of $\tilde{\LCd}_{\mathrm{KG}}$.
This yields with the derivative Eq.~(\ref{eq:metric-identity2}) of the determinant $g$ of the covariant metric $g_{\mu\nu}$ with respect to the contravariant metric:
\allowdisplaybreaks
\begin{align*}
\pfrac{\tilde{\LCd}_{\mathrm{KG}}}{g^{\nu\beta}}g^{\mu\beta}&+\pfrac{\tilde{\LCd}_{\mathrm{KG}}}{g^{\alpha\nu}}g^{\alpha\mu}\\
&=\left(\pfrac{\bar{\phi}}{x^{\nu}}\pfrac{\phi}{x^{\beta}}
+\pfrac{\bar{\phi}}{x^{\beta}}\pfrac{\phi}{x^{\nu}}\right)g^{\mu\beta}\sqrt{-g}-\delta_\nu^\mu \tilde{\LCd}_{\mathrm{KG}}\\
\pfrac{\bar{\phi}}{x^\nu}\pfrac{\tilde{\LCd}_{\mathrm{KG}}}{\left(\pfrac{\bar{\phi}}{x^\mu}\right)}
&=\frac{1}{2}\pfrac{\bar{\phi}}{x^{\nu}}\left(\pfrac{\phi}{x^{\beta}}g^{\mu\beta}+\pfrac{\phi}{x^{\alpha}}g^{\alpha\mu}\right)\sqrt{-g}\\
&=\pfrac{\bar{\phi}}{x^{\nu}}\pfrac{\phi}{x^{\beta}}g^{\mu\beta}\sqrt{-g}\\
\pfrac{\tilde{\LCd}_{\mathrm{KG}}}{\left(\pfrac{\phi}{x^\mu}\right)}\pfrac{\phi}{x^\nu}
&=\frac{1}{2}\left(\pfrac{\bar{\phi}}{x^{\alpha}}g^{\alpha\mu}+\pfrac{\bar{\phi}}{x^{\beta}}g^{\mu\beta}\right)\pfrac{\phi}{x^{\nu}}\sqrt{-g}\\
&=\pfrac{\bar{\phi}}{x^{\beta}}\pfrac{\phi}{x^{\nu}}g^{\mu\beta}\sqrt{-g}.
\end{align*}
\allowdisplaybreaks[0]
The particular identity for the Lagrangian density~(\ref{eq:KG-Lagrangian-complex}) thus writes:
\begin{equation}
-\pfrac{\tilde{\LCd}_{\mathrm{KG}}}{g^{\nu\beta}}g^{\mu\beta}-\pfrac{\tilde{\LCd}_{\mathrm{KG}}}{g^{\beta\nu}}g^{\beta\mu}
+\pfrac{\bar{\phi}}{x^\nu}\pfrac{\tilde{\LCd}_{\mathrm{KG}}}{\left(\pfrac{\bar{\phi}}{x^\mu}\right)}
+\pfrac{\tilde{\LCd}_{\mathrm{KG}}}{\left(\pfrac{\phi}{x^\mu}\right)}\pfrac{\phi}{x^\nu}\equiv\delta_\nu^\mu \tilde{\LCd}_{\mathrm{KG}}
\end{equation}
It can be expressed equivalently as the identity for the system's metric and canonical energy-momentum tensor densities,
$\tilde{T}\indices{^\mu_\nu}$ and $\tilde{\theta}\indices{^\mu_\nu}$:
\begin{equation}
\tilde{T}\indices{^\mu_\nu}\!\!=2\pfrac{\tilde{\LCd}_{\mathrm{KG}}}{g^{\nu\beta}}g^{\mu\beta}\equiv
\pfrac{\bar{\phi}}{x^\nu}\pfrac{\tilde{\LCd}_{\mathrm{KG}}}{\left(\pfrac{\bar{\phi}}{x^\mu}\right)}
+\pfrac{\tilde{\LCd}_{\mathrm{KG}}}{\left(\pfrac{\phi}{x^\mu}\right)}\pfrac{\phi}{x^\nu}-\delta_\nu^\mu \tilde{\LCd}_{\mathrm{KG}}\!=\tilde{\theta}\indices{^\mu_\nu}
\end{equation}
\subsection{Proca Lagrangian\label{ex:proca-ham}}
In a dynamic space-time, the Proca Lagrangian is defined by (see e.g.~\cite{greiner96}):
\begin{equation*}
\tilde{\LCd}_{\mathrm{P}}=\left(-\quarter f_{\alpha\beta} f_{\xi\eta} g^{\alpha\xi} g^{\beta\eta}
+\onehalf m^{2}a_{\alpha} a_{\beta} g^{\alpha\beta}\right)\sqrt{-g},
\end{equation*}
with
\begin{equation*}
f_{\alpha\beta}\equiv\pfrac{a_\beta}{x^\alpha}-\pfrac{a_\alpha}{x^\beta}.
\end{equation*}
The scalar density $\tilde{\LCd}_{\mathrm{P}}$ represents in this case a \emph{quadratic}
function of the field tensor $f_{\alpha\beta}$ and the field vector $a_\alpha$.
Nonetheless, the identity~(\ref{eq:assertion3}) for the scalar density $\tilde{\LCd}_{\mathrm{P}}$ follows in the usual form as:
\begin{equation}\label{eq:proca-identity}
-\pfrac{\tilde{\LCd}_{\mathrm{P}}}{g^{\nu\beta}}g^{\mu\beta}-\pfrac{\tilde{\LCd}_{\mathrm{P}}}{g^{\beta\nu}}g^{\beta\mu}
+\pfrac{\tilde{\LCd}_{\mathrm{P}}}{f_{\mu\beta}}f_{\nu\beta}
+\pfrac{\tilde{\LCd}_{\mathrm{P}}}{f_{\beta\mu}}f_{\beta\nu}+\pfrac{\tilde{\LCd}_{\mathrm{P}}}{a_\mu}a_\nu\equiv\delta_\nu^\mu\tilde{\LCd}_{\mathrm{P}}.
\end{equation}
Rewriting the identity~(\ref{eq:proca-identity}) as a correlation of metric and canonical energy-momentum tensor densities,
$\tilde{T}\indices{^\mu_\nu}$ and $\tilde{\theta}\indices{^\mu_\nu}$,
we observe that these tensors do \emph{not} coincide for the Proca system:
\begin{align}
\tilde{T}\indices{^\mu_\nu}=2\pfrac{\tilde{\LCd}_{\mathrm{P}}}{g^{\beta\nu}}g^{\beta\mu}
&\equiv\pfrac{\tilde{\LCd}_{\mathrm{P}}}{f_{\beta\mu}}f_{\beta\nu}-\delta_\nu^\mu\tilde{\LCd}_{\mathrm{P}}
+\pfrac{\tilde{\LCd}_{\mathrm{P}}}{f_{\mu\beta}}f_{\nu\beta}+\pfrac{\tilde{\LCd}_{\mathrm{P}}}{a_\mu}a_\nu\nonumber\\
&\equiv\tilde{\theta}\indices{^\mu_\nu}+\pfrac{\tilde{\LCd}_{\mathrm{P}}}{f_{\mu\beta}}f_{\nu\beta}+\pfrac{\tilde{\LCd}_{\mathrm{P}}}{a_\mu}a_\nu.
\end{align}
\subsection{Regularized Dirac Lagrangian}\label{sec:dirac-all}
\subsubsection{Identity for the metric indices}\label{sec:dirac-metr}
The regularized Dirac Lagrangian $\LCd_{\mathrm{D}}\big(\psi,\partial\psi,\bar{\psi},\partial\bar{\psi},\bgamma^\mu\big)$ \citep{gasiorowicz66,struckmeier08,struckmeier21a},
constructed upon the conventional Dirac equation writes
\begin{equation}\label{ld-regular}
\LCd_{\mathrm{D}}=\frac{\rmi}{2}\left(\bar{\psi} \bgamma^{\alpha}\pfrac{\psi}{x^{\alpha}}
-\pfrac{\bar{\psi}}{x^{\alpha}}\bgamma^{\alpha}\psi\right)-m \bar{\psi}\psi
+\frac{\rmi}{3M}\pfrac{\bar{\psi}}{x^{\alpha}} \bsigma^{\alpha\beta}\pfrac{\psi}{x^{\beta}},
\end{equation}
wherein the $(1,1)$-spinor-$(2,0)$-tensor field $\bsigma^{\alpha\beta}$ is defined
as the \emph{commutator} of the matrix product $\bgamma^{\alpha}\bgamma^{\nu}$:
\begin{equation}\label{eq:sigma-def}
\bsigma^{\alpha\nu}=\frac{\rmi}{2}\left(\bgamma^{\alpha}\bgamma^{\nu}-\bgamma^{\nu}\bgamma^{\alpha}\right)
\;\Leftrightarrow\;\sigma\indices{^{a}_{b}^{\alpha\nu}}=\frac{\rmi}{2}
\left(\gamma\indices{^{a}_{c}^{\alpha}} \gamma\indices{^{c}_{b}^{\nu}}-
\gamma\indices{^{a}_{c}^{\nu}} \gamma\indices{^{c}_{b}^{\alpha}}\right).
\end{equation}
With an explicit notation of the spinor indices as lower case Latin letters,
the Dirac Lagrangian takes on the form:
\begin{align}
\LCd_{\mathrm{D}}=\frac{\rmi}{2}&\left(\bar{\psi}_a \gamma\indices{^a_c^\alpha}\pfrac{\psi^c}{x^{\alpha}}
-\pfrac{\bar{\psi}_a}{x^{\alpha}}\gamma\indices{^a_c^\alpha}\psi^c\right)-m \bar{\psi}_a\psi^a\nonumber\\
&+\frac{\rmi}{3M}\pfrac{\bar{\psi}_a}{x^{\alpha}} \sigma\indices{^a_c^{\alpha\beta}}\pfrac{\psi^c}{x^{\beta}}.
\label{ld-regular-expl}
\end{align}
With the spinor indices fully contracted, the identity~(\ref{eq:assertion1}) with respect to
the metric indices for the scalar quantity $\LCd_{\mathrm{D}}$ is then given by:
\begin{equation}\label{ld-identity}
\pfrac{\LCd_{\mathrm{D}}}{\left(\pfrac{\psi^c}{x^\mu}\right)}\pfrac{\psi^c}{x^\nu}
+\pfrac{\bar{\psi}_a}{x^\nu}\pfrac{\LCd_{\mathrm{D}}}{\left(\pfrac{\bar{\psi}_a}{x^\mu}\right)}
-\pfrac{\LCd_{\mathrm{D}}}{\gamma\indices{^a_c^\nu}}\gamma\indices{^a_c^\mu}\equiv0.
\end{equation}
Equation~(\ref{ld-identity}) is verified by direct calculation:
\begin{align*}
\pfrac{\LCd_{\mathrm{D}}}{\left(\pfrac{\psi^c}{x^\mu}\right)}\pfrac{\psi^c}{x^\nu}
&=\frac{\rmi}{2}\bar{\psi}_a \gamma\indices{^a_c^\mu}\pfrac{\psi^c}{x^{\nu}}
+\frac{\rmi}{3M}\pfrac{\bar{\psi}_a}{x^{\alpha}} \sigma\indices{^a_c^{\alpha\mu}}\pfrac{\psi^c}{x^{\nu}}\\
\pfrac{\bar{\psi}_a}{x^\nu}\pfrac{\LCd_{\mathrm{D}}}{\left(\pfrac{\bar{\psi}_a}{x^\mu}\right)}
&=-\frac{\rmi}{2}\pfrac{\bar{\psi}_a}{x^\nu}\gamma\indices{^a_c^\mu}\psi^c
+\frac{\rmi}{3M}\pfrac{\bar{\psi}_a}{x^\nu}\sigma\indices{^a_c^{\mu\beta}}\pfrac{\psi^c}{x^{\beta}}\\
-\pfrac{\LCd_{\mathrm{D}}}{\gamma\indices{^a_c^\nu}}\gamma\indices{^a_c^\mu}
&=-\frac{\rmi}{2}\left(\bar{\psi}_a \gamma\indices{^a_c^\mu}\pfrac{\psi^c}{x^{\nu}}
-\pfrac{\bar{\psi}_a}{x^{\nu}}\gamma\indices{^a_c^\mu}\psi^c\right)\\
&\quad-\frac{\rmi}{3M}\left(\pfrac{\bar{\psi}_a}{x^{\nu}} \sigma\indices{^a_c^{\mu\beta}}\pfrac{\psi^c}{x^{\beta}}
+\pfrac{\bar{\psi}_a}{x^{\alpha}} \sigma\indices{^a_c^{\alpha\mu}}\pfrac{\psi^c}{x^{\nu}}\right),
\end{align*}
which obviously sums up to zero.

The identity~(\ref{eq:assertion3}) for the covariant scalar \emph{density} $\tilde{\LCd}_{\mathrm{D}}$ is then
\begin{equation*}
\pfrac{\tilde{\LCd}_{\mathrm{D}}}{\left(\pfrac{\psi^c}{x^\mu}\right)}\pfrac{\psi^c}{x^\nu}
+\pfrac{\bar{\psi}_a}{x^\nu}\pfrac{\tilde{\LCd}_{\mathrm{D}}}{\left(\pfrac{\bar{\psi}_a}{x^\mu}\right)}
-\pfrac{\tilde{\LCd}_{\mathrm{D}}}{\gamma\indices{^a_c^\nu}}\gamma\indices{^a_c^\mu}\equiv\delta_\nu^\mu \tilde{\LCd}_{\mathrm{D}},
\end{equation*}
hence, rearranging the terms and skipping the spinor indices:
\begin{align}
\tilde{\theta}\indices{^\mu_\nu}&\equiv
\pfrac{\tilde{\LCd}_{\mathrm{D}}}{\left(\pfrac{\psi}{x^\mu}\right)}\pfrac{\psi}{x^\nu}
+\pfrac{\bar{\psi}}{x^\nu}\pfrac{\tilde{\LCd}_{\mathrm{D}}}{\left(\pfrac{\bar{\psi}}{x^\mu}\right)}
-\delta_\nu^\mu \tilde{\LCd}_{\mathrm{D}}\equiv
\Tr\left\{\pfrac{\tilde{\LCd}_{\mathrm{D}}}{\bgamma\indices{^\nu}}\bgamma\indices{^\mu}\right\}\nonumber\\
&\equiv\tilde{T}\indices{^\mu_\nu}.
\label{ld-identity2}
\end{align}
Equation~(\ref{ld-identity2}) states that the \emph{canonical} energy-momentum tensor density $\tilde{\theta}\indices{^\mu_\nu}$
coincides for the Dirac Lagrangian with the \emph{metric} energy-momentum tensor $\tilde{T}\indices{^\mu_\nu}$.
The explicit covariant representation of these tensors is not symmetric:
\begin{align}
T\indices{_\mu_\nu}=\theta\indices{_\mu_\nu}
&=\frac{\rmi}{2}\left(\bar{\psi} \bgamma\indices{_\mu}\pfrac{\psi}{x^{\nu}}
-\pfrac{\bar{\psi}}{x^{\nu}}\bgamma\indices{_\mu}\psi\right)-g_{\mu\nu}\LCd_{\mathrm{D}}\nonumber\\
&\quad+\frac{\rmi}{3M}g_{\mu\alpha}\left(\pfrac{\bar{\psi}}{x^{\nu}} \bsigma\indices{^\alpha^\beta}\pfrac{\psi}{x^{\beta}}
+\pfrac{\bar{\psi}}{x^{\beta}} \bsigma\indices{^\beta^\alpha}\pfrac{\psi}{x^{\nu}}\right).
\label{dirac-emt}
\end{align}
and thus also provides a source for torsion of space-time in a generalized Einstein-type equation.
\section{Theories of Gravity}\label{sec:gravity}
In the following we consider a generalized, metric compatible space-time with Riemann-Cartan geometry,
where the connection is not necessarily symmetric and torsion of space-time is thus admitted.
The additional degrees of freedom carried by the torsion enforce to treat connection and metric as independent fields (Palatini approach).
The metric and torsion dependent curvature tensor is then called the Riemann-Cartan tensor.
In some cases we will also, in addition, take non-metricity into account, or also restrict
the geometry to the torsion-free Riemannian one with the Levi-Civita connection.
\subsection{Scalars depending on the Riemann-Cartan tensor and the metric}
\subsubsection{General gravity Lagrangian for metric compatibility}
Any \emph{relative} scalar Lagrangian $\tilde{\LCd}_\mathrm{Gr}(R,g)$ of weight $w=1$ built from the Riemann-Cartan curvature tensor
\begin{equation}\label{eq:riemann-tensor}
R\indices{^{\eta}_{\xi\beta\lambda}}=\pfrac{\Gamma\indices{^{\eta}_{\xi\lambda}}}{x^{\beta}}-
\pfrac{\Gamma\indices{^{\eta}_{\xi\beta}}}{x^{\lambda}}+
\Gamma\indices{^{\eta}_{\tau\beta}}\Gamma\indices{^{\tau}_{\xi\lambda}}-
\Gamma\indices{^{\eta}_{\tau\lambda}}\Gamma\indices{^{\tau}_{\xi\beta}}
\end{equation}
and the metric satisfies the identity~(\ref{eq:assertion3}), whose particular form for $\tilde{\LCd}_\mathrm{Gr}$ is given by:
\begin{align}
\pfrac{\tilde{\LCd}_\mathrm{Gr}}{g^{\nu\beta}}g^{\mu\beta}&+\pfrac{\tilde{\LCd}_\mathrm{Gr}}{g^{\beta\nu}}g^{\beta\mu}\equiv
-\pfrac{\tilde{\LCd}_\mathrm{Gr}}{R\indices{^{\nu}_{\alpha\beta\lambda}}}R\indices{^{\mu}_{\alpha\beta\lambda}}
+\pfrac{\tilde{\LCd}_\mathrm{Gr}}{R\indices{^{\eta}_{\mu\beta\lambda}}}R\indices{^{\eta}_{\nu\beta\lambda}}\nonumber\\
&\quad+\pfrac{\tilde{\LCd}_\mathrm{Gr}}{R\indices{^{\eta}_{\alpha\mu\lambda}}}R\indices{^{\eta}_{\alpha\nu\lambda}}
+\pfrac{\tilde{\LCd}_\mathrm{Gr}}{R\indices{^{\eta}_{\alpha\beta\mu}}}R\indices{^{\eta}_{\alpha\beta\nu}}-\delta_\nu^\mu \tilde{\LCd}_\mathrm{Gr}.
\label{eq:constraint2-lag}
\end{align}
As stated in Eq.~(\ref{eq:assertion2}) it is necessary and sufficient for the identity~(\ref{eq:constraint2-lag})
to hold that $\tilde{\LCd}_\mathrm{Gr}(R,g)$ is a scalar density function of the \emph{tensors} $R\indices{^{\nu}_{\alpha\beta\lambda}}$ and $g^{\beta\nu}$.
Whether or not a tensor may on its part be represented by underlying quantities---such as the affine connection or the metric
in the actual case---does not play any role for setting up the identity.
The fact that the Riemann-Cartan tensor actually possesses internal symmetries also does not modify the form of the identity.
Clearly, these symmetries will be reflected in a later analysis of the invariant.
For instance, with the Riemann-Cartan tensor~(\ref{eq:riemann-tensor}) being by its definition skew-symmetric in its last index pair,
the last two derivative terms of the identity~(\ref{eq:constraint2-lag}) are identical and can be merged into a single term.

The left-hand side of the identity~(\ref{eq:constraint2-lag}) can again be interpreted as the \emph{metric}
energy-momentum tensor $\tilde{T}\indices{^\mu_\nu}$ of the Lagrangian system $\tilde{\LCd}_\mathrm{Gr}(R,g)$, hence
\begin{equation}\label{eq:emt-metr-gen}
T\indices{^\mu_\nu}=\frac{1}{\sqrt{-g}}\left(\pfrac{\tilde{\LCd}_\mathrm{Gr}}{g^{\nu\beta}}g^{\mu\beta}+\pfrac{\tilde{\LCd}_\mathrm{Gr}}{g^{\beta\nu}}g^{\beta\mu}\right),
\end{equation}
whereas the corresponding \emph{canonical} energy-momentum tensor $\theta\indices{^\mu_\nu}$ is derived
for the corresponding \emph{absolute} scalar $\LCd_\mathrm{Gr}=\tilde{\LCd}_\mathrm{Gr}/\sqrt{-g}$ by means of Noether's theorem~\citep{struckmeier18a} as:
\begin{equation}\label{eq:emt-can-gen}
\vartheta\indices{^\mu_\nu}=2\pfrac{\LCd_\mathrm{Gr}}{R\indices{^{\eta}_{\alpha\beta\mu}}}R\indices{^{\eta}_{\alpha\beta\nu}}-\delta_\nu^\mu \LCd_\mathrm{Gr}.
\end{equation}
Notice that we use a modified notation for the canonical energy-momentum tensor for gravity to distinguish it for later convenience from that of matter.
This gives the relation between both tensors:
\begin{equation}\label{eq:emt-can-rel}
T\indices{^\mu_\nu}=\vartheta\indices{^\mu_\nu}
-R\indices{^{\mu}_{\alpha\beta\lambda}}\pfrac{\LCd_\mathrm{Gr}}{R\indices{^{\nu}_{\alpha\beta\lambda}}}
+\pfrac{\LCd_\mathrm{Gr}}{R\indices{^{\eta}_{\mu\beta\lambda}}}R\indices{^{\eta}_{\nu\beta\lambda}}.
\end{equation}
We remark that the above tensors were formally set up
in analogy to the energy-momentum tensors of the Klein-Gordon and Proca systems.
This is independent from interpreting any of the tensors~(\ref{eq:emt-metr-gen}) and~(\ref{eq:emt-can-gen})
as the \emph{physical} energy-momentum tensors of the gravitational field.
\subsubsection{General gravity Lagrangian}
For the case of a non-vanishing covariant derivative of the metric, hence for \mbox{$g\indices{^{\alpha\beta}_{;\mu}}\not\equiv0$},
the gravity Lagrangian also depends on the covariant derivative of the metric $\tilde{\LCd}_\mathrm{Gr}=\tilde{\LCd}_\mathrm{Gr}(R,g,\nabla g)$.
The identity~(\ref{eq:constraint2-lag}) is then amended by three additional terms that reflect the three indices of $g\indices{^\mu^\nu_{;\alpha}}$:
\begin{align}
&\pfrac{\tilde{\LCd}_\mathrm{Gr}}{g\indices{^\alpha^\beta_{;\mu}}}g\indices{^\alpha^\beta_{;\nu}}\!\!
-\pfrac{\tilde{\LCd}_\mathrm{Gr}}{g\indices{^\nu^\beta_{;\alpha}}}g\indices{^\mu^\beta_{;\alpha}}\!\!
-\pfrac{\tilde{\LCd}_\mathrm{Gr}}{g\indices{^\beta^\nu_{;\alpha}}}g\indices{^\beta^\mu_{;\alpha}}\!\!
-\pfrac{\tilde{\LCd}_\mathrm{Gr}}{g^{\nu\beta}}g^{\mu\beta}\!-\pfrac{\tilde{\LCd}_\mathrm{Gr}}{g^{\beta\nu}}g^{\beta\mu}\nonumber\\
&-\pfrac{\tilde{\LCd}_\mathrm{Gr}}{R\indices{^{\nu}_{\alpha\beta\lambda}}}\!R\indices{^{\mu}_{\alpha\beta\lambda}}\!
+\pfrac{\tilde{\LCd}_\mathrm{Gr}}{R\indices{^{\eta}_{\mu\beta\lambda}}}\!R\indices{^{\eta}_{\nu\beta\lambda}}\!
+\pfrac{\tilde{\LCd}_\mathrm{Gr}}{R\indices{^{\eta}_{\alpha\mu\lambda}}}\!R\indices{^{\eta}_{\alpha\nu\lambda}}\!
+\pfrac{\tilde{\LCd}_\mathrm{Gr}}{R\indices{^{\eta}_{\alpha\beta\mu}}}\!R\indices{^{\eta}_{\alpha\beta\nu}}\nonumber\\
&\equiv\delta_\nu^\mu \tilde{\LCd}_\mathrm{Gr}.
\label{eq:constraint2a-lag}
\end{align}
The \emph{metric} energy-momentum tensor is then generalized to:
\begin{align*}
T\indices{^\mu_\nu}=\frac{1}{\sqrt{-g}}&\left(\pfrac{\tilde{\LCd}_\mathrm{Gr}}{g^{\nu\beta}}g^{\mu\beta}+\pfrac{\tilde{\LCd}_\mathrm{Gr}}{g^{\beta\nu}}g^{\beta\mu}\right.\\
&+\pfrac{\tilde{\LCd}_\mathrm{Gr}}{g\indices{^{\nu\beta}_{;\alpha}}}g\indices{^{\mu\beta}_{;\alpha}}
+\pfrac{\tilde{\LCd}_\mathrm{Gr}}{g\indices{^{\beta\nu}_{;\alpha}}}g\indices{^{\beta\mu}_{;\alpha}}\bigg),
\end{align*}
whereas the \emph{canonical} energy-momentum tensor takes on the generalized form~(see Eq.~(48) of Struckmeier~et~al.~\cite{struckmeier18a}):
\begin{equation*}
\vartheta\indices{^\mu_\nu}=\pfrac{\LCd_\mathrm{Gr}}{g\indices{^{\alpha\beta}_{;\mu}}}g\indices{^{\alpha\beta}_{;\nu}}+
\pfrac{\LCd_\mathrm{Gr}}{R\indices{^{\eta}_{\alpha\mu\lambda}}}R\indices{^{\eta}_{\alpha\nu\lambda}}+
\pfrac{\LCd_\mathrm{Gr}}{R\indices{^{\eta}_{\alpha\beta\mu}}}R\indices{^{\eta}_{\alpha\beta\nu}}-\delta_\nu^\mu \LCd_\mathrm{Gr}.
\end{equation*}
From the identity~(\ref{eq:assertion3}), we again recover Eq.~(\ref{eq:emt-can-rel}) for the correlation of
the so-defined energy-momentum tensors.
Obviously, the correlation~(\ref{eq:emt-can-rel}) of these tensors holds independently of metric compatibility.
\subsubsection{Einstein-Hilbert-Cartan Lagrangian \texorpdfstring{$\tilde{\LCd}_{E}$}{}}\label{sec:ehc-Lag}
The Einstein-Hilbert-Cartan Lagrangian $\tilde{\LCd}_{E}$ is defined by the Ricci scalar density $\tilde{R}=R\sqrt{-g}$.
The latter emerges from the definition of the Riemann-Cartan curvature tensor from Eq.~(\ref{eq:riemann-tensor})
via its non-trivial contraction $R_{\xi\lambda}=R\indices{^{\eta}_{\xi\eta\lambda}}$,
followed by a contraction with the metric according to:
\begin{equation}\label{eq:ricci-scalar}
16\pi G \tilde{\LCd}_{E}=R\indices{^{\eta}_{\xi\eta\lambda}} g^{\xi\lambda} \sqrt{-g}=R_{\xi\lambda} g^{\xi\lambda} \sqrt{-g}=\tilde{R}.
\end{equation}
Note that the Ricci tensor $R_{\xi\lambda}$ is \emph{not} symmetric for non-symmetric connection coefficients
$\Gamma\indices{^{\eta}_{\xi\lambda}}\neq\Gamma\indices{^{\eta}_{\lambda\xi}}$.
For the scalar density $\tilde{\LCd}_{E}$, the left-hand side of the general equation~(\ref{eq:constraint2-lag}),
hence the derivatives of $\tilde{R}$ with respect to the metric are:
\begin{subequations}\label{eq:Hilbert-deris}
\begin{equation}\label{eq:Hilbert-deri1}
\pfrac{\tilde{R}}{g^{\nu\beta}}g^{\mu\beta}+\pfrac{\tilde{R}}{g^{\beta\nu}}g^{\beta\mu}
=\tilde{R}\indices{_\nu^\mu}+\tilde{R}\indices{^\mu_\nu}-\delta_\nu^\mu \tilde{R}.
\end{equation}
The derivatives of $\tilde{R}$ with respect to the Riemann tensor follow as:
\begin{align}
\pfrac{\tilde{R}}{R\indices{^\nu_{\xi\beta\lambda}}}R\indices{^\mu_{\xi\beta\lambda}}
&=\delta^\eta_\nu\delta^\beta_\eta g^{\xi\lambda}\tilde{R}\indices{^\mu_{\xi\beta\lambda}}
=\tilde{R}\indices{^\mu^\lambda_{\nu\lambda}}=\tilde{R}\indices{^\mu_\nu}\label{eq:Hilbert-deri2}\\
\pfrac{\tilde{R}}{R\indices{^\eta_{\mu\beta\lambda}}}R\indices{^\eta_{\nu\beta\lambda}}
&=\delta_\xi^\mu\delta_\eta^\beta g^{\xi\lambda}\tilde{R}\indices{^\eta_{\nu\beta\lambda}}
=\tilde{R}\indices{^\beta_\nu_\beta^\mu}=\tilde{R}\indices{_\nu^\mu}\label{eq:Hilbert-deri3}\\
\pfrac{\tilde{R}}{R\indices{^\eta_{\xi\mu\lambda}}}R\indices{^\eta_{\xi\nu\lambda}}
&=\delta_\eta^\mu\hphantom{\delta_\eta^\beta} g^{\xi\lambda}\tilde{R}\indices{^\eta_{\xi\nu\lambda}}
=\tilde{R}\indices{^\mu^\lambda_\nu_\lambda}=\tilde{R}\indices{^\mu_\nu}\label{eq:Hilbert-deri4}\\
\pfrac{\tilde{R}}{R\indices{^\eta_{\xi\beta\mu}}}R\indices{^\eta_{\xi\beta\nu}}
&=\delta_\lambda^\mu\delta_\eta^\beta g^{\xi\lambda}\tilde{R}\indices{^\eta_{\xi\beta\nu}}
=\tilde{R}\indices{^\eta^\mu_\eta_\nu}=\tilde{R}\indices{^\mu_\nu}\label{eq:Hilbert-deri5}.
\end{align}
\end{subequations}
The identity~(\ref{eq:constraint2-lag}) is obviously satisfied by Eqs.~(\ref{eq:Hilbert-deris}).
The derivatives of the $\tilde{\LCd}_{E}$ with respect to the metric
define the \emph{metric} energy-momentum tensor,
\begin{align}
T\indices{^\mu_\nu}&=\frac{1}{\sqrt{-g}}\left(\pfrac{\tilde{\LCd}_{E}}{g^{\nu\beta}}g^{\mu\beta}
+\pfrac{\tilde{\LCd}_{E}}{g^{\beta\nu}}g^{\beta\mu}\right)\nonumber\\
&=\frac{1}{16\pi G}\left(R\indices{_\nu^\mu}+R\indices{^\mu_\nu}-\delta_\nu^\mu R\right).
\label{eq:hilbert-emt-met}
\end{align}
By virtue of the identity~(\ref{eq:constraint2-lag}), this tensor can be identically replaced by its derivatives with respect to the Riemann tensor.
The corresponding \emph{canonical} energy-momentum tensor $\vartheta\indices{^\mu_\nu}$ is defined by:
\begin{align}
\vartheta\indices{^\mu_\nu}&=\pfrac{\LCd_{E}}{R\indices{^\eta_{\xi\mu\lambda}}}R\indices{^\eta_{\xi\nu\lambda}}
+\pfrac{\LCd_{E}}{R\indices{^\eta_{\xi\beta\mu}}}R\indices{^\eta_{\xi\beta\nu}}-\delta_\nu^\mu\LCd_{E}\nonumber\\
&=\frac{1}{16\pi G}\left(2R\indices{^\mu_\nu}-\delta_\nu^\mu R\right).
\label{eq:hilbert-emt-can}
\end{align}
Both tensors are thus related as:
\begin{equation*}
T\indices{^\mu_\nu}=\vartheta\indices{^\mu_\nu}+\frac{1}{16\pi G}\left(R\indices{_\nu^\mu}-R\indices{^\mu_\nu}\right)
\end{equation*}
and hence agree for a symmetric Ricci tensor $R_{\mu\nu}\equiv R_{\nu\mu}$, which is induced by symmetric connections
$\Gamma\indices{^{\eta}_{\xi\lambda}}\equiv\Gamma\indices{^{\eta}_{\lambda\xi}}$ and hence a torsion-free space-time.

Allowing for an asymmetric connection, Eq.~(\ref{eq:hilbert-emt-can}) yields an asymmetric canonical energy-momentum tensor.
In their covariant representations both energy-momentum tensors are:
\begin{align}
T_{\nu\mu}&=\frac{1}{8\pi G}\left(R_{(\nu\mu)}-\onehalf g_{\nu\mu}R\right)\equiv T_{(\nu\mu)}\label{eq:emt-met-hilbert}\\
\vartheta_{\nu\mu}&=\frac{1}{8\pi G}\left(R_{\nu\mu}-\onehalf g_{\nu\mu}R\right),
\end{align}
which gives the relations
\begin{equation}\label{eq:emt-can-hilbert}
\vartheta_{(\nu\mu)}=T_{\nu\mu},\qquad\vartheta_{[\nu\mu]}=\frac{1}{8\pi G} R_{[\nu\mu]}.
\end{equation}
Note that these derivations hold for both, the metric and the Palatini formulation as both are based on a covariant scalar density Lagrangian.
\subsubsection{Vierbein formulation of the Einstein-Hilbert-Cartan Lagrangian}\label{sec:hilbert-tetr}
For including fermions to the variety of matter fields the description of space-time geometry
must be modified by considering inertial frames at every point of the underlying Riemann-Cartan geometry.
These frames are provided by a ``vierbein'' (aka tetrad) field of four orthonormal basis vectors $\pmb{e}_i(x)$
with the coordinates $e\indices{_i^\mu}(x)$.
The Lorentz (or nonholonomic) indices $i$ refer to the local inertial (Lorentz) frame while $\mu$ are the metric (holonomic) indices.
Dual vierbeins, $\pmb{e}^i(x)$, satisfy the orthonormality relation $\pmb{e}^i \pmb{e}_j = \eta_{ij}$ with the Minkowski metric $\eta_{ij}$.
The Lagrangian~(\ref{eq:ricci-scalar}) may now be expressed equivalently in terms of the vierbeins as
\begin{equation}\label{eq:ricci-scalar-tetr}
16\pi G \tilde{\LCd}_{E}=R\indices{^k_{i\xi\lambda}} e\indices{_k^\xi} e\indices{_j^\lambda} \eta^{ij} \dete=R \dete=\tilde{R},
\end{equation}
where $\dete\equiv\sqrt{-g}$ is the determinant of the dual vierbein field $e\indices{^i_\lambda}$.
The identity~(\ref{eq:constraint2-lag}) then separates into the two index classes, namely the metric index class (holonomic indices)
\begin{equation}\label{eq:Ricci-identity1}
-\pfrac{\tilde{\LCd}_{E}}{e\indices{_m^\nu}}e\indices{_m^\mu}
+\pfrac{\tilde{\LCd}_{E}}{R\indices{^n_{m\mu\lambda}}}R\indices{^n_{m\nu\lambda}}
+\pfrac{\tilde{\LCd}_{E}}{R\indices{^n_{m\beta\mu}}}R\indices{^n_{m\beta\nu}}
\equiv\delta_\nu^\mu\tilde{\LCd}_{E},
\end{equation}
and the Lorentz index class (nonholonomic indices) of the vierbeins:
\begin{align}
&-\pfrac{\tilde{\LCd}_{E}}{e\indices{_n^\beta}}e\indices{_m^\beta}
+\pfrac{\tilde{\LCd}_{E}}{\eta\indices{^m^k}}\eta\indices{^n^k}
+\pfrac{\tilde{\LCd}_{E}}{\eta\indices{^k^m}}\eta\indices{^k^n}\nonumber\\
&+\pfrac{\tilde{\LCd}_{E}}{R\indices{^m_{k\beta\lambda}}}R\indices{^n_{k\beta\lambda}}
-\pfrac{\tilde{\LCd}_{E}}{R\indices{^k_{n\beta\lambda}}}R\indices{^k_{m\beta\lambda}}
\equiv\delta_m^n\tilde{\LCd}_{E}.
\label{eq:Ricci-identity2}
\end{align}
To confirm Eqs.~(\ref{eq:Ricci-identity1}) and~(\ref{eq:Ricci-identity2}), we use
\begin{equation*}
\pfrac{\dete}{e\indices{_i^\nu}}=-e\indices{^i_\nu} \dete,
\end{equation*}
to find for the terms with open metric indices and contracted vierbein indices:
\begin{align*}
\pfrac{\tilde{R}}{e\indices{_m^\nu}}e\indices{_m^\mu}
&=R\indices{^k_{i\xi\lambda}}\left(\delta_\nu^\xi \delta_k^m e\indices{_j^\lambda}+e\indices{_k^\xi} \delta_\nu^\lambda \delta_j^m
-e\indices{_k^\xi} e\indices{_j^\lambda} e\indices{^m_\nu}\right)e\indices{_m^\mu} \eta^{ij} \dete\\
&=\left(R\indices{^m_{i\nu\lambda}} e\indices{_j^\lambda} e\indices{_m^\mu}
+R\indices{^k_{i\xi\nu}} e\indices{_k^\xi} e\indices{_j^\mu}
-\delta_\nu^\mu R\indices{^k_{i\xi\lambda}} e\indices{_k^\xi} e\indices{_j^\lambda}\right)\eta^{ij} \dete\\
&=\left(2R\indices{^\mu_\nu}-\delta_\nu^\mu R\right)\dete
\end{align*}
\begin{align*}
\pfrac{\tilde{R}}{R\indices{^n_{m\mu\lambda}}}R\indices{^n_{m\nu\lambda}}
&=\delta_n^k \delta_i^m\delta_\xi^\mu e\indices{_k^\xi} e\indices{_j^\lambda}
R\indices{^n_{m\nu\lambda}} \eta^{ij} \dete
=R\indices{^\mu_\nu} \dete\\
\pfrac{\tilde{R}}{R\indices{^n_{m\beta\mu}}}R\indices{^n_{m\beta\nu}}
&=\delta_n^k \delta_i^m \delta_\xi^\beta \delta_\lambda^\mu e\indices{_k^\xi} e\indices{_j^\lambda}
R\indices{^n_{m\beta\nu}} \eta^{ij} \dete
=R\indices{^\mu_\nu} \dete,
\end{align*}
which sums up to yield the right-hand side of~(\ref{eq:Ricci-identity1}).
With the canonical energy-momentum tensor density corresponding to Eq.~(\ref{eq:emt-can-gen}), which is obtained here from Eq.~(\ref{eq:Ricci-identity1}) as
\begin{equation*}
\tilde{\vartheta}\indices{^\mu_\nu}=\pfrac{\tilde{\LCd}_{E}}{R\indices{^n_{m\mu\lambda}}}R\indices{^n_{m\nu\lambda}}
+\pfrac{\tilde{\LCd}_{E}}{R\indices{^n_{m\beta\mu}}}R\indices{^n_{m\beta\nu}}-\delta_\nu^\mu\tilde{\LCd}_{E},
\end{equation*}
the ``vierbein'' energy-momentum tensor $\bar{T}\indices{^\mu_\nu}$ is not symmetric and coincides with the canonical one:
\begin{equation*}
\bar{T}\indices{^\mu_\nu}=\frac{1}{\dete}\pfrac{\tilde{\LCd}_{E}}{e\indices{_m^\nu}}e\indices{_m^\mu}
=\vartheta\indices{^\mu_\nu}=\frac{1}{8\pi G}\left(R\indices{^\mu_\nu}-\onehalf\delta_\nu^\mu R\right).
\end{equation*}
In the same way, we set up the terms with contracted metric indices while leaving the vierbein indices open:
\begin{align*}
-\pfrac{\tilde{R}}{e\indices{_n^\tau}}e\indices{_m^\tau}
&=-R\indices{^k_{i\xi\lambda}}\left(\delta_\tau^\xi \delta_k^n e\indices{_j^\lambda}+e\indices{_k^\xi} \delta^\lambda_\tau \delta^n_j
-e\indices{_k^\xi} e\indices{_j^\lambda} e\indices{^n_\tau}\right)e\indices{_m^\tau} \eta^{ij} \dete\\
&=\left(\delta_m^n R-2R\indices{^n_m}\right)\dete
\end{align*}
\begin{align*}
\pfrac{\tilde{R}}{\eta\indices{^m^k}}\eta\indices{^n^k}+\pfrac{\tilde{R}}{\eta\indices{^k^m}}\eta\indices{^k^n}
&=\left(R\indices{^n_m}+R\indices{_m^n}\right)\dete\\
\pfrac{\tilde{R}}{R\indices{^m_{k\beta\lambda}}}R\indices{^n_{k\beta\lambda}}
&=\delta_m^l \delta_i^k \delta_\xi^\beta e\indices{_l^\xi} e\indices{_j^\lambda}
R\indices{^n_{k\beta\lambda}} \eta^{ij} \dete=R\indices{^n_m} \dete\\
-\pfrac{\tilde{R}}{R\indices{^k_{n\beta\lambda}}}R\indices{^k_{m\beta\lambda}}
&=-\delta_k^l \delta_i^n \delta_\xi^\beta e\indices{_l^\xi} e\indices{_j^\lambda}
R\indices{^k_{m\beta\lambda}} \eta^{ij} \dete=-R\indices{_m^n} \dete,
\end{align*}
which sums up to $\delta_m^n \tilde{R}$ and thus verifies the identity~(\ref{eq:Ricci-identity2}).
\subsubsection{Quadratic gravity with Riemann tensor squared}
The gravity Lagrangian $\LCd_{\mathrm{Rie}^2}$ constructed form the complete contraction of two Riemann-Cartan tensors,
also referred to as the Kretschmann scalar, was already proposed by Einstein in a private letter to H.~Weyl~\citep{einstein18}.
It is defined by:
\begin{equation*}
\tilde{\LCd}_{\mathrm{Rie}^2}=-\quarter R\indices{^\xi_{\eta\rho\alpha}}R\indices{^\eta_{\xi\sigma\lambda}} g^{\rho\sigma} g^{\alpha\lambda}\sqrt{-g}.
\end{equation*}
It is again directly verified that the left-hand side of the identity~(\ref{eq:constraint2-lag}) gives
\begin{subequations}\label{eq:riem-deri-all}
\begin{equation}
\pfrac{\tilde{\LCd}_{\mathrm{Rie}^2}}{g^{\nu\beta}}g^{\mu\beta}+\pfrac{\tilde{\LCd}_{\mathrm{Rie}^2}}{g^{\beta\nu}}g^{\beta\mu}
=R\indices{_{\eta\rho\tau}^\mu}R\indices{^{\eta\rho\tau}_\nu}\sqrt{-g}-\delta_\nu^\mu \tilde{\LCd}_{\mathrm{Rie}^2},
\end{equation}
while the terms on the right-hand side evaluate to
\allowdisplaybreaks
\begin{align}
-\pfrac{\tilde{\LCd}_{\mathrm{Rie}^2}}{R\indices{^\nu_{\alpha\beta\lambda}}}R\indices{^\mu_{\alpha\beta\lambda}}
&=-\onehalf R\indices{^\mu_{\alpha\beta\lambda}} R\indices{_\nu^{\alpha\beta\lambda}}\sqrt{-g}\label{eq:riem-deri-1}\\
\pfrac{\tilde{\LCd}_{\mathrm{Rie}^2}}{R\indices{^\eta_{\mu\beta\lambda}}}R\indices{^\eta_{\nu\beta\lambda}}
&=\hphantom{-}\onehalf R\indices{^\mu_{\alpha\beta\lambda}} R\indices{_\nu^{\alpha\beta\lambda}}\sqrt{-g}\label{eq:riem-deri-2}\\
\pfrac{\tilde{\LCd}_{\mathrm{Rie}^2}}{R\indices{^\eta_{\alpha\mu\lambda}}}R\indices{^\eta_{\alpha\nu\lambda}}
&=\hphantom{-}\onehalf R\indices{^\eta_{\alpha\beta}^\mu} R\indices{_\eta^{\alpha\beta}_\nu}\sqrt{-g}\label{eq:riem-deri-3}\\
\pfrac{\tilde{\LCd}_{\mathrm{Rie}^2}}{R\indices{^\eta_{\alpha\beta\mu}}}R\indices{^\eta_{\alpha\beta\nu}}
&=\hphantom{-}\onehalf R\indices{^\eta_{\alpha\beta}^\mu} R\indices{_\eta^{\alpha\beta}_\nu}\sqrt{-g}\label{eq:riem-deri-4}.
\end{align}
\end{subequations}
\allowdisplaybreaks[0]
The terms~(\ref{eq:riem-deri-all}) indeed satisfy the general form of the identity~(\ref{eq:constraint2-lag}).
As the terms~(\ref{eq:riem-deri-1}) and~(\ref{eq:riem-deri-2}) cancel, the remaining terms~(\ref{eq:riem-deri-3}) and~(\ref{eq:riem-deri-4})
form the canonical energy-momentum tensor, which is symmetric as it agrees with the metric one:
\begin{align*}
T_{\nu\mu}&=\frac{2}{\sqrt{-g}}\pfrac{\tilde{\LCd}_{\mathrm{Rie}^2}}{g^{\nu\mu}}
=R\indices{_{\eta\rho\tau\mu}}R\indices{^{\eta\rho\tau}_\nu}-g_{\nu\mu}\LCd_{\mathrm{Rie}^2}\\
\vartheta_{\nu\mu}&=
2\pfrac{\tilde{\LCd}_{\mathrm{Rie}^2}}{R\indices{^\eta_{\alpha\beta}^\mu}}R\indices{^\eta_{\alpha\beta\nu}}-g_{\nu\mu}\LCd_{\mathrm{Rie}^2}=T_{\nu\mu}.
\end{align*}
Consequently, a pure $\LCd_{\mathrm{Rie}^2}$ model for the dynamics of the free gravitational field leads to a \emph{symmetric}
canonical energy-momentum tensor.
It thus does not contribute to torsion dynamics of space-time, as described by Eq.~(\ref{eq:emt-can-hilbert}).
\subsubsection{Quadratic gravity with Ricci tensor squared}
The scalar density made of the square of the (not necessarily symmetric) Ricci tensor $R_{\eta\alpha}=R\indices{^\rho_{\eta\rho\alpha}}$
is defined by the following contraction with the metric
\begin{equation}\label{eq:ricci-tensor-scalar}
\tilde{\LCd}_{\mathrm{Ric}^2}=\onehalf R\indices{_{\eta\alpha}}R\indices{_{\xi\lambda}} g^{\eta\xi} g^{\alpha\lambda}\sqrt{-g}.
\end{equation}
For the Lagrangian~(\ref{eq:ricci-tensor-scalar}), the general identity~(\ref{eq:constraint2-lag}) reduces to:
\begin{equation}\label{eq:identity-ricci}
-\pfrac{\tilde{\LCd}_{\mathrm{Ric}^2}}{g^{\nu\beta}}g^{\mu\beta}-\pfrac{\tilde{\LCd}_{\mathrm{Ric}^2}}{g^{\beta\nu}}g^{\beta\mu}
+\pfrac{\tilde{\LCd}_{\mathrm{Ric}^2}}{R\indices{_{\mu\beta}}}R\indices{_{\nu\beta}}
+\pfrac{\tilde{\LCd}_{\mathrm{Ric}^2}}{R\indices{_{\beta\mu}}}R\indices{_{\beta\nu}}
\equiv\delta_\nu^\mu \tilde{\LCd}_{\mathrm{Ric}^2}.
\end{equation}
Explicitly, the derivative terms with respect to the metric are
\begin{subequations}\label{eq:ric-deri-all}
\begin{equation}
\pfrac{\tilde{\LCd}_{\mathrm{Ric}^2}}{g^{\nu\beta}}g^{\mu\beta}+\pfrac{\tilde{\LCd}_{\mathrm{Ric}^2}}{g^{\beta\nu}}g^{\beta\mu}
=\left(R\indices{_{\nu\beta}}R\indices{^{\mu\beta}}+R\indices{_{\beta\nu}}R\indices{^{\beta\mu}}\right)\!\sqrt{-g}-\delta_\nu^\mu \tilde{\LCd}_{\mathrm{Ric}^2}
\end{equation}
whereas the derivative terms with respect to the Ricci tensor follow as
\begin{align}
\pfrac{\tilde{\LCd}_{\mathrm{Ric}^2}}{R\indices{_{\mu\beta}}}R\indices{_{\nu\beta}}
&=R\indices{^{\mu\beta}} R\indices{_{\nu\beta}}\sqrt{-g}\label{eq:ric-deri-2}\\
\pfrac{\tilde{\LCd}_{\mathrm{Ric}^2}}{R\indices{_{\beta\mu}}}R\indices{_{\beta\nu}}
&=R\indices{^{\beta\mu}} R\indices{_{\beta\nu}}\sqrt{-g}\label{eq:ric-deri-3}.
\end{align}
\end{subequations}
Clearly, the right-hand sides of Eqs.~(\ref{eq:ric-deri-all}) satisfy Eq.~(\ref{eq:identity-ricci}).
The metric energy-momentum tensor is thus:
\begin{equation*}
T\indices{^\mu_\nu}=\frac{2}{\sqrt{-g}}\pfrac{\tilde{\LCd}_{\mathrm{Ric}^2}}{g^{\nu\beta}}g^{\mu\beta}=R\indices{^\mu^\beta}R\indices{_{\nu\beta}}
+R\indices{^\beta^\mu}R\indices{_{\beta\nu}}-\delta_\nu^\mu \LCd_{\mathrm{Ric}^2}.
\end{equation*}
Since the canonical energy-momentum tensor emerges as:
\begin{align*}
\vartheta\indices{^\mu_\nu}&=\pfrac{\LCd_{\mathrm{Ric}^2}}{R\indices{_{\mu\beta}}}R\indices{_{\nu\beta}}
+\pfrac{\LCd_{\mathrm{Ric}^2}}{R\indices{_{\beta\mu}}}R\indices{_{\beta\nu}}-\delta_\nu^\mu \LCd_{\mathrm{Ric}^2}\\
&=R\indices{^{\mu\beta}} R\indices{_{\nu\beta}}+R\indices{^{\beta\mu}} R\indices{_{\beta\nu}}-\delta_\nu^\mu \LCd_{\mathrm{Ric}^2},
\end{align*}
it coincides with the metric one.
$\vartheta\indices{_\mu_\nu}$ is thus symmetric in this case even if the Ricci tensor itself is non-symmetric.
As a consequence, it cannot act as a counterpart of a non-symmetric matter field energy-momentum tensor in an Einstein-type equation.
\subsubsection{Quadratic gravity with Ricci scalar squared}
The scalar density made of the square of the Ricci scalar $R=R_{\eta\alpha}g^{\eta\alpha}=R\indices{^\rho_{\eta\rho\alpha}}g^{\eta\alpha}$
is defined by the following contraction with the metric
\begin{equation}\label{eq:ricci-scalar-scalar}
\tilde{\LCd}_{\mathrm{RiS}^2}=\onehalf R\indices{_{\eta\alpha}}R\indices{_{\xi\lambda}} g^{\eta\alpha} g^{\xi\lambda}\sqrt{-g}.
\end{equation}
The identity~(\ref{eq:identity-ricci}) consisting of the derivative terms with respect to the metric
\begin{align*}
\pfrac{\tilde{\LCd}_{\mathrm{RiS}^2}}{g^{\nu\beta}}g^{\mu\beta}+\pfrac{\tilde{\LCd}_{\mathrm{RiS}^2}}{g^{\beta\nu}}g^{\beta\mu}
&=R\indices{_\beta^\beta}\left(R\indices{_\nu^\mu}+R\indices{^\mu_\nu}\right)\sqrt{-g}-\delta_\nu^\mu \tilde{\LCd}_{\mathrm{RiS}^2},
\end{align*}
and the derivative terms with respect to the Ricci tensor
\begin{subequations}
\begin{align}
\pfrac{\tilde{\LCd}_{\mathrm{RiS}^2}}{R\indices{_{\mu\beta}}}R\indices{_{\nu\beta}}
&=R\indices{_\beta^\beta} R\indices{_\nu^\mu}\sqrt{-g}\label{eq:ris-deri-2}\\
\pfrac{\tilde{\LCd}_{\mathrm{RiS}^2}}{R\indices{_{\beta\mu}}}R\indices{_{\beta\nu}}
&=R\indices{_\beta^\beta} R\indices{^\mu_\nu}\sqrt{-g}\label{eq:ris-deri-3}
\end{align}
\end{subequations}
is obviously satisfied.
The metric energy-momentum tensor is thus given by:
\begin{equation*}
T\indices{^\mu_\nu}=\frac{2}{\sqrt{-g}}\pfrac{\tilde{\LCd}_{\mathrm{RiS}^2}}{g^{\nu\beta}}g^{\mu\beta}=R\indices{_\beta^\beta}\left(R\indices{_\nu^\mu}
+R\indices{^\mu_\nu}\right)-\delta_\nu^\mu \LCd_{\mathrm{RiS}^2},
\end{equation*}
whereas the canonical energy-momentum tensor emerges as:
\begin{align*}
\vartheta\indices{^\mu_\nu}&=\pfrac{\LCd_{\mathrm{RiS}^2}}{R\indices{_{\mu\beta}}}R\indices{_{\nu\beta}}
+\pfrac{\LCd_{\mathrm{RiS}^2}}{R\indices{_{\beta\mu}}}R\indices{_{\beta\nu}}-\delta_\nu^\mu \LCd_{\mathrm{RiS}^2}\\
&=R\indices{_\beta^\beta}\left(R\indices{_\nu^\mu}+R\indices{^\mu_\nu}\right)-\delta_\nu^\mu \LCd_{\mathrm{RiS}^2}.
\end{align*}
Also in this case, the canonical energy-momentum tensor coincides with the metric one.
\subsection{Consistency equation of the gauge theory of gravity}
From the covariant canonical Hamiltonian formulation of the gauge theory of gravity (CCGG)~\citep{struckmeier17a}
one derives the second rank tensor ``consistency equation'':
\begin{align}
&\hphantom{=}-2\pfrac{\tilde{\HCd}_\mathrm{Gr}}{g_{\alpha\mu}}g_{\alpha\nu}+
2\tilde{k}^{ \alpha\mu\beta}\pfrac{\tilde{\HCd}_\mathrm{Gr}}{\tilde{k}\indices{^{\alpha\nu\beta}}}-
\tilde{q}\indices{_{\nu}^{\eta\alpha\beta}}\pfrac{\tilde{\HCd}_\mathrm{Gr}}{\tilde{q}\indices{_{\mu}^{\eta\alpha\beta}}}+
\tilde{q}\indices{_{\eta}^{\mu\alpha\beta}}\pfrac{\tilde{\HCd}_\mathrm{Gr}}{\tilde{q}\indices{_{\eta}^{\nu\alpha\beta}}}\nonumber\\
&=\pfrac{\tilde{\HCd}_{0}}{a_{\mu}}a_{\nu}-\tilde{p}^{ \mu\beta}\pfrac{\tilde{\HCd}_{0}}{\tilde{p}^{ \nu\beta}}+
2\pfrac{\tilde{\HCd}_{0}}{g_{\alpha\mu}}g_{\alpha\nu},
\label{eq:consistency2}
\end{align}
which has the Lagrangian representation:
\begin{align}
&-2\pfrac{\tilde{\LCd}_\mathrm{Gr}}{g^{\alpha\nu}}g^{\alpha\mu}+
2\pfrac{\tilde{\LCd}_\mathrm{Gr}}{g_{\alpha\mu;\beta}}g_{\alpha\nu;\beta}-
\pfrac{\tilde{\LCd}_\mathrm{Gr}}{R\indices{^{\nu}_{\eta\alpha\beta}}}R\indices{^{\mu}_{\eta\alpha\beta}}+
\pfrac{\tilde{\LCd}_\mathrm{Gr}}{R\indices{^{\eta}_{\mu\alpha\beta}}}R\indices{^{\eta}_{\nu\alpha\beta}}\nonumber\\
&=-\pfrac{\tilde{\LCd}_{0}}{a_{\mu}}a_{\nu}-\pfrac{\tilde{\LCd}_{0}}{a_{\mu;\beta}}a_{\nu;\beta}+
2\pfrac{\tilde{\LCd}_{0}}{g^{\alpha\nu}}g^{\alpha\mu}.
\label{eq:consistency2a}
\end{align}
$\tilde{\LCd}_\mathrm{Gr}$ represents a generic Lagrangian density of vacuum gravity, and $\tilde{\LCd}_0$
can be seen as representing scalar and massive vector fields, i.e. $\tilde{\LCd}_0 = \tilde{\LCd}_\mathrm{KG}+\tilde{\LCd}_\mathrm{P}$.
By virtue of the \emph{identities} for scalar density-valued functions of arbitrary tensors and the metric, given here by
\begin{align*}
\delta_{\nu}^{\mu}\tilde{\LCd}_{0}&\equiv
\pfrac{\tilde{\LCd}_{0}}{\left(\pfrac{\phi}{x^{\mu}}\right)}\pfrac{\phi}{x^{\nu}}
+\pfrac{\tilde{\LCd}_{0}}{a_{\mu}}a_{\nu}
+\pfrac{\tilde{\LCd}_{0}}{a_{\mu;\beta}}a_{\nu;\beta}+\pfrac{\tilde{\LCd}_{0}}{a_{\beta;\mu}}a_{\beta;\nu}\nonumber\\
&\hphantom{\equiv}-2\pfrac{\tilde{\LCd}_{0}}{g^{\alpha\nu}}g^{\alpha\mu}\nonumber\\
\delta_{\nu}^{\mu}\tilde{\LCd}_\mathrm{Gr}&\equiv
-2\pfrac{\tilde{\LCd}_\mathrm{Gr}}{g^{\alpha\nu}}g^{\alpha\mu}+
2\pfrac{\tilde{\LCd}_\mathrm{Gr}}{g_{\alpha\mu;\beta}}g_{\alpha\nu;\beta}
+\pfrac{\tilde{\LCd}_\mathrm{Gr}}{g_{\alpha\beta;\mu}}g_{\alpha\beta;\nu}\nonumber\\
&\hphantom{\equiv}-\pfrac{\tilde{\LCd}_\mathrm{Gr}}{R\indices{^{\nu}_{\eta\alpha\beta}}}R\indices{^{\mu}_{\eta\alpha\beta}}
+\pfrac{\tilde{\LCd}_\mathrm{Gr}}{R\indices{^{\eta}_{\mu\alpha\beta}}}R\indices{^{\eta}_{\nu\alpha\beta}}
+2\pfrac{\tilde{\LCd}_\mathrm{Gr}}{R\indices{^{\eta}_{\alpha\beta\mu}}}R\indices{^{\eta}_{\alpha\beta\nu}},
\end{align*}
the consistency equation~(\ref{eq:consistency2a}) has the equivalent representation:
\begin{align}
&\hphantom{=}\pfrac{\tilde{\LCd}_\mathrm{Gr}}{g_{\alpha\beta;\mu}}g_{\alpha\beta;\nu}
+2\pfrac{\tilde{\LCd}_\mathrm{Gr}}{R\indices{^\eta_{\alpha\beta\mu}}}R\indices{^\eta_{\alpha\beta\nu}}-\delta_\nu^\mu\tilde{\LCd}_\mathrm{Gr}\nonumber\\
&=-\Bigg(\pfrac{\tilde{\LCd}_0}{\left(\pfrac{\phi}{x^\mu}\right)}\pfrac{\phi}{x^\nu}
+\pfrac{\tilde{\LCd}_0}{a_{\beta;\mu}}a_{\beta;\nu}-\delta_\nu^\mu\tilde{\LCd}_0\Bigg).
\label{eq:consistency3}
\end{align}
The right-hand side of Eq.~(\ref{eq:consistency3}) is exactly minus the covariant representation of the
non-symmetric \emph{canonical} energy-momentum tensor density $\tilde{\theta}\indices{^\mu_\nu}$ of the system $\tilde{\LCd}_0$:
\begin{equation*}
\tilde{\theta}\indices{^\mu_\nu}=\pfrac{\tilde{\LCd}_{0}}{\left(\pfrac{\phi}{x^\mu}\right)}\pfrac{\phi}{x^\nu}
+\pfrac{\tilde{\LCd}_{0}}{a_{\beta;\mu}}a_{\beta;\nu}-\delta_\nu^\mu\tilde{\LCd}_{0}.
\end{equation*}
We note that the covariant derivatives of the vector field term causes the connection
besides the metric to act as an additional coupling to the left-hand side of Eq.~(\ref{eq:consistency3}).
If the latter is interpreted as the \emph{canonical} energy-momentum tensor density
associated with the ``free'' gravitational field Lagrangian $\tilde{\LCd}_\mathrm{Gr}$,
then, Eq.~(\ref{eq:consistency3}) is equivalently written as
\begin{equation}\label{eq:consistency4}
\pfrac{\LCd_\mathrm{Gr}}{g_{\alpha\beta;\mu}}g_{\alpha\beta;\nu}+
2\pfrac{\LCd_\mathrm{Gr}}{R\indices{^\eta_{\alpha\beta\mu}}}R\indices{^\eta_{\alpha\beta\nu}}
-\delta_\nu^\mu\LCd_\mathrm{Gr}=-\theta\indices{^\mu_\nu}.
\end{equation}
This is the general Einstein-type equation that applies for any postulated ``free gravity'' Lagrangian $\LCd_\mathrm{Gr}$
describing the dynamics of the gravitational field in source-free regions of space-time including torsion and non-metricity.
Moreover, written as
\begin{equation}
\vartheta\indices{^\mu_\nu}+\theta\indices{^\mu_\nu}=0
\end{equation}
this represents a balance equation of energy and momentum of matter and space-time, often expressed as the
\emph{zero-energy universe}~\citep{lorentz1916,levi-civita1917,jordan39,sciama53,feynman62,hawking03}.

For the common case of metricity, hence for a covariantly conserved metric ($g_{\alpha\beta;\nu}\equiv0$), Eq.~(\ref{eq:consistency4}) reduces to:
\begin{equation}\label{eq:consistency5}\boxed{
2\pfrac{\LCd_\mathrm{Gr}}{R\indices{^\eta_{\alpha\beta\mu}}}R\indices{^\eta_{\alpha\beta}^\nu}
-g^{\mu\nu}\LCd_\mathrm{Gr}=-\theta\indices{^\mu^\nu}.}
\end{equation}
For the Einstein-Hilbert-Cartan Lagrangian $\LCd_{\mathrm{Gr,E}}=R/16\pi G$, discussed in Sect.~\ref{sec:ehc-Lag},
which constitutes a particular model Lagrangian to describe the dynamics of the free gravitational field,
Eq.~(\ref{eq:consistency5}) immediately yields the Einstein equation.
It also holds for non-zero torsion in the case of a non-symmetric matter field tensor $\theta\indices{^\mu^\nu}$,
which requires the Ricci tensor $R\indices{^\mu^\nu}$ to be non-symmetric as well:
\begin{align*}
R^{(\mu\nu)}-\onehalf g^{\mu\nu}R&=-8\pi G \theta^{(\mu\nu)}\\
R^{[\mu\nu]}&=-8\pi G \theta^{[\mu\nu]}.
\end{align*}
The first equation is the classical Einstein equation.
The second equation follows  from Eq.~(\ref{eq:emt-can-hilbert}) and can be expressed in terms of the skew-symmetric portion
of the affine connection which is the \emph{torsion} $S\indices{^{\lambda}_{\mu\nu}}$ of space-time:
\begin{equation*}
\left[\left(\nabla_\lambda-2S_\lambda\right)S\indices{^\lambda_\mu_\nu}
-\left(\nabla_\mu S_\nu-\nabla_\nu S_\mu\right)\right]=-{8\pi G} \theta_{[\mu\nu]}.
\end{equation*}
In the last equation we used the Bianchi identity relating torsion and the antisymmetric Ricci tensor.
Obviously, a non-symmetric canonical energy-momentum tensor emerges as a source of {torsion} of space-time.
\section{Conclusions}\label{sec:conclusions}
The theorem on (relative) scalar-valued functions of tensors reflects the linear structure of tensor spaces
and thereby constitutes an analogue to Euler's theorem on homogeneous functions.
When applied to Lagrangian densities $\tilde{\LCd}$ of classical field theories, the resulting identities
provide relations between the metric and canonical versions of the energy-momentum tensors of matter fields.
By the same mathematical reasoning, analogous tensors for generic theories of gravity are derived.
The identity thereby sheds new light on the long-standing ambiguity of the proper definition of the energy-momentum of gravity.
Moreover, as a further application of the identity to gauge theories of gravity (see, e.g.~\cite{struckmeier17a}),
a relation of the energy-momentum tensors of both, gravity and matter, emerges.
In combination with the identities that hold separately for both tensors, that relation may be cast into a
simple but nonetheless most general form of an energy-momentum balance equation that confirms the conjecture of ``zero-energy universe''.
The symmetric portion of that balance equation reproduces thereby a generalized version of Einstein's field equation,
while the skew-symmetric portion gives a Poisson-type equation for the torsion of space-time.


\bibliography{identity6}



\end{document}